\documentclass[11pt, a4paper]{amsart}
\usepackage{amsfonts}
\usepackage{amssymb}
\usepackage{amsmath}
\usepackage[dvips]{graphics}
\usepackage{epsfig}
\usepackage[section]{placeins}
\usepackage{euscript}
\usepackage{color}
\usepackage{xcolor}
\usepackage{tikz} \usetikzlibrary{patterns}
\usepackage{cancel}
\usepackage{fancyvrb}

\usepackage[pdftex]{hyperref}

\definecolor{edgeblue}{RGB}{21,101,192}
\definecolor{bandgrey}{RGB}{78,78,84}
\definecolor{gammapurple}{RGB}{106,27,120}

 \newtheorem{thm}{Theorem}[section]
 \newtheorem{cor}[thm]{Corollary}
 \newtheorem{lemma}[thm]{Lemma}
 \newtheorem{prop}[thm]{Proposition}
 \theoremstyle{definition}
 \newtheorem{defn}[thm]{Definition}
 \newtheorem{rem}[thm]{Remark}

 \newtheorem{assumption}[thm]{Assumption}

 \numberwithin{equation}{section}

\newcommand{\vb}[1]{\boldsymbol{#1}}
\newcommand{\ket}[1]{|{#1}\rangle}
\newcommand{\bra}[1]{\langle{#1}|}

\newcommand{\Zee}{\mathbb{Z}}
\newcommand{\Cee}{\mathbb{C}}
\newcommand{\Tee}{\mathbb{T}}
\newcommand{\real}{\mathbb{R}}
\newcommand{\ft}[1]{\widehat{#1}}

\newcommand{\bbC}{{\mathbb C}}

\newcommand{\bbF}{{\mathbb F}}

\newcommand{\bbN}{{\mathbb N}}

\newcommand{\bbT}{{\mathbb T}}

\newcommand{\bbZ}{{\mathbb Z}}

\newcommand{\Ker}[1]{\mathrm{Ker} #1 }

\newcommand{\str}{^{*}}

\newcommand{\tr}{\mathrm{tr}}

\newcommand{\teop}{\hfill$\square$}

\newcommand{\be}{\begin{equation}}
\newcommand{\ee}{\end{equation}}
\newcommand{\bea}{\begin{eqnarray}}
\newcommand{\eea}{\end{eqnarray}}
\newcommand{\beann}{\begin{eqnarray*}}
\newcommand{\eeann}{\end{eqnarray*}}

\begin{document}

\title{Free Fermion frustration-free models have zero Hall conductance}

\author{Sven Bachmann}
\address{Department of Mathematics, The University of British Columbia, 1984 Mathematics Road, Vancouver, BC, V6T 1Z2 Canada}
\email{sbach@math.ubc.ca}

\author{Zhiqian Du}
\address{Department of Mathematics, University of California Davis, One Shields Ave, Davis, CA 95616}
\email{duzhiqian@gmail.com}

\author{Martin Fraas}
\address{Department of Mathematics, University of California Davis, One Shields Ave, Davis, CA 95616}
\email{mfraas@ucdavis.edu}

\date{\today}

\begin{abstract}
We give a new proof that translation-invariant free Fermion frustration-free models have zero Hall conductance. Our proof is analytic in nature and it uses the bulk-edge correspondence. The central argument, which holds beyond the translation-invariant situation, is that in frustration-free models, edge modes cannot merge into the ground state band. We give an example of a gapped frustration-free model that is nonetheless gapless on the half-plane.
\end{abstract}

\maketitle

\section{Introduction}
The quantum Hall effect remains the paradigmatic experimental example of a topological phase of matter. The quantization and stability of quantum Hall conductance are explained by linking it to a topological invariant of a ground state phase \cite{TKNN, A.S.S.ChernNumber, HastingsMichalakis, bachmann2020many}. Furthermore, the fractional quantum Hall effect has anyonic excitations \cite{Frohlich1, Frohlich2, Frohlich3}. The braiding and fusion properties of these anyons have the structure of a modular tensor category, providing another topological invariant \cite{bachmann2025tensor}. In fact, it is conjectured that the modular tensor category describing anyonic excitations, together with the chiral central charge (which is a generalization of the Hall conductance) classify all 2D phases of matter \cite{Wen_Classification}. Yet, the theory of fractional quantum Hall effect contains a puzzle -- we do not know a single microscopic model where we can prove that the ground state has fractional Hall conductance. This should be contrasted with a zoo of exactly solvable models \cite{levin2005string}, see also \cite{BolsKjaer2025} that realize modular tensor categories of any quantum double.

We believe that a resolution of this puzzle is that all topological exactly solvable models that we know are finite range and frustration-free. In fact, usually they are not only frustration-free but the local Hamiltonians are also commuting. The Hall conductance of a frustration-free ground state of a commuting Hamiltonian model is zero \cite{kapustin2020local, zhang2022vanishing}. It is a natural conjecture that, in fact, the Hall conductance in any finite range frustration-free model is zero. We do not know how to prove this conjecture in general and we focus on the case of free Fermion models. We show
\begin{thm}
The Hall conductance of translation-invariant free Fermion frustration-free ground state vanishes. 
\end{thm}
A formal version of the theorem is Theorem~\ref{thm:edge}. The formal version separates the edge and the bulk index. Under some generic conditions, it is well-known fact~\cite{GrafPorta} that the edge index, see Definition~\ref{def:Edge Index}, and the bulk index are equal. That Hall conductance vanishes in these systems is known. It was first proved in \cite{ChenTrivalFlatBand} with alternative proofs given later in \cite{Dubail_2015, read2017compactly}. The latter paper \cite{read2017compactly} also shows that the $\mathbb{Z}_2$ index is trivial. These works study the bulk index and use techniques of homological algebra and $K$-theory. In contrast, we study the edge index and give an analytical proof that only uses Kato's perturbation theory. Using the edge picture, the proof of triviality of $\mathbb{Z}_2$ index is exactly the same as the proof of vanishing of the Hall conductance. 

Let us elaborate a little more about the previous works. In \cite{ChenTrivalFlatBand} it is shown that flat bands have zero Chern number. A flat ground state band has frustration-free parent Hamiltonian and vice versa, see Proposition~\ref{prop:flat_band}. So the results of \cite{ChenTrivalFlatBand} imply our statement but are somewhat more general because they also show triviality of bands that are not ground state bands. The article \cite{Dubail_2015} gives a no-go theorem for the existence of Gaussian tensor network representations of the ground state. Such tensor network representations imply the existence of frustration-free parent Hamiltonian. Finally, \cite{read2017compactly} studies the existence of compactly supported Wannier functions, and it includes an excellent discussion of how this implies and generalizes the previous results. In addition, we should mention that~\cite{sathe2025topological} considers Hamiltonians, possibly many-body, whose energy levels are all flat, though the exact connection to our work is less clear. The reader might notice that these works and our work have all mutually distinct motivations and framing. We believe that this shows the richness of the problem.

As we were finishing our work, we learned about the series of works by Masaoka, Ono and co-authors \cite{masaoka2026frustration, ono2026frustration}. They study general theory of frustration-free free Fermion systems and the discussion in Section~\ref{sec:FFFF} overlaps with their work. In fact, their work is more general as they only assume that the many-body Hamiltonian is frustration-free after subtracting the ground state energy, while we always assume that this energy is zero.

The key to our method is an almost trivial observation that for frustration-free edge Hamiltonians, there are no edge currents along an edge that respects the frustration-free decomposition. This holds for both single-body and many-body Hamiltonians; however by itself it does not prove the statement because Hall conductance is proportional to the variation of the edge current with chemical potential, not to the edge current itself. Changing the chemical potential might destroy the frustration-free decomposition. In translation invariant free Fermion systems, edge currents are related to the existence of edge states, and we prove our result by showing that these edge states cannot emerge from the ground state band.

A natural question is whether the edge states can touch the ground state band. We give an example showing that this can indeed happen, providing a frustration-free model that is gapped in the bulk, has vanishing Hall conductance but nonetheless has a gapless boundary. The only previous example was given in \cite{bachmann2014product}, but these so-called PVBS states have gapless boundaries only at a very specific angle of the edge with respect to the square lattice.

In the following section we explain our setting and the main results. The section concludes with a number of remarks on the time-reversal invariant case, comparisons with explicit models, and the robustness of the edge argument, in particular its validity for disordered and interacting models. In Section~\ref{sec:edge} we analyze the edge states and prove our main result. In Section~\ref{sec:example} we give an example of a free Fermion frustration-free Hamiltonian with gapless boundary. Finally, for completeness, we reproduce in Section~\ref{sec:bulk} the proof that the Chern number of the ground state band is zero using the techniques of \cite{Dubail_2015}.


\section{Frustration-free free Fermions}
\label{sec:FFFF}
We consider the lattice $\mathbb{Z}^2$ equipped with the $\ell^\infty$ metric, and denote $\vb{e}_1 = (1,0), \, \vb{e}_2 = (0,1)$. The single particle Hilbert space is $\mathcal{H} = \ell^2(\mathbb{Z}^2, \mathbb{C}^N)$ for some $N\in \mathbb{N}$. The components of $\Cee^N$ are called spin components. For ${\vb{x}} \in \mathbb{Z}^2$, let $\delta_{\vb{x}}$ denote the normalized function that vanishes everywhere but on site ${\vb{x}}$. We consider a finite range Hamiltonian $H$, meaning that the $N\times N$-matrix $\langle \delta_{\vb{x}},\, H \,\delta_{\vb{y}}\rangle$ vanishes if $|{\vb{x}} -{\vb{y}}|$ is larger than some constant. After second quantization, the Hilbert space is the fermionic Fock space $\mathcal{F} = \Gamma(\mathcal{H})$ and the second quantized Hamiltonian\footnote{Since this part of the text only serves the purpose of motivating the single particle setting, we ignore technical details on the precise meaning to give to these formal second quantized expressions.} is
$$
d \Gamma(H) = \sum_{{\vb{x}},{\vb{y}},j,k} H_{{\vb{x}},{\vb{y}},j,k} {a}_{{\vb{x}},j}^* {a}_{{\vb{y}},k},
$$
where ${a}_{{\vb{x}},j}$ is an annihilation operator on site ${\vb{x}}$ with on-site spin components labeled by $j$. The Hamiltonian is called frustration-free if there exists a decomposition 
$$
d \Gamma(H) = \sum_{{\vb{x}}} \Phi_{\vb{x}}
$$
such that (i) the interactions are non-negative $\Phi_{\vb{x}} \geq 0$, (ii) $\Phi_{\vb{x}}$ have uniformly finite range  and (iii) the kernels of $\{\Phi_{\vb{x}}:\vb{x}\in\mathbb{Z}^2\}$ have nonzero intersection. In the context of free Fermion systems we will also assume that $\Phi_{\vb{x}} = d \Gamma(H_{\vb{x}})$ for some single particle Hamiltonian $H_{\vb{x}}$. With this, the frustration-free condition can be rephrased in terms of the single particle Hamiltonian.
\begin{assumption}[F$^4$ Hamiltonian]
A self-adjoint operator $H$ on $\mathcal{H}$ is called a frustration-free free Fermion Hamiltonian if there exists a decomposition $H = \sum_{{\vb{x}} \in \mathbb{Z}^2} H_{\vb{x}}$ such that:
\begin{enumerate}
    \label{AssumptionF4}
\item For all ${\vb{x}} \in \mathbb{Z}^2$, $H_{\vb{x}} \geq 0$,
\item There exists $r>0$, such that $H_{\vb{x}} \delta_{\vb{y}}=0 $ if $|{\vb{x}}- \vb{y}| > r$,
\item $\cap_{{\vb{x}} \in \mathbb{Z}^2} \mathrm{Ker} H_{\vb{x}} \setminus \{0\} \neq \emptyset$.
\end{enumerate}
We note that given i., iii.~is equivalent to $H$ possessing a ground state with zero energy.
\end{assumption}
To connect to the many-body framework, let $P$ be the orthogonal projection onto the subspace $ \Ker{H} = \cap_{{\vb{x}} \in \mathbb{Z}^2} \mathrm{Ker} H_{\vb{x}}$. Then $P$ or any nonzero sub-projection of $P$ corresponds to a quasi-free ground state of $d \Gamma(H)$. The only scenario that will be interesting for us is when $P$ is infinite dimensional, in which case $d \Gamma(H)$ has infinitely many ground states. This is in particular the case in the translation-invariant setting, see~below. Specifically, we will consider the `fully occupied' ground state, namely the one corresponding to the full projection $P$. In conjunction with the following gap assumption this is equivalent to having the Fermi energy in the gap between zero and the rest of the spectrum.
\begin{assumption}[Spectral gap]
There exists $\gamma >0$ such that $\sigma(H) \subseteq \{0\} \cup (\gamma, \infty)$.
\end{assumption}
Notice that $\gamma$ is not unique for a gapped $H$, and we may use any $\gamma$ within an interval for the convenience of argument.\\
It is important to note that in the F$^4$ setting there is no natural way, at least not known to the authors, how to adjust the chemical potential. Indeed, the transformation $d \Gamma(H) \to d \Gamma(H) -  \mu \sum_{{\vb{x}},j} a_{{\vb{x}},j}^* a_{{\vb{x}},j}$ corresponds to $ H \to H - \mu$, and $H - \mu$ does not satisfy the F$^4$ condition if $H$ does. The new Hamiltonian is not positive with zero in the spectrum, and even if this condition could be relaxed, there is no reason why $H - \mu$ would possess a frustration-free decomposition. In view of this, one should rather think of the assumptions as selecting a chemical potential for which the Hamiltonian is frustration-free and gapped, but this might be a single point and the F$^4$ property is not stable under perturbations.

The conditions i.~and ii.~of the F$^4$ assumption imply that there is $M\in\mathbb{N}$ such that for any ${\vb{x}} \in \mathbb{Z}^2$, there exists $M_{\vb{x}}\in\mathbb{N}$ with $M_{\vb{x}} \leq M$, and compactly supported vectors $\ket{v_{j,{\vb{x}}}}$, $j=1, \dots, M_{\vb{x}}$,  such that
\begin{equation*}
H_{\vb{x}} = \sum_{j=1}^{M_{\vb{x}}} \ket{v_{j,{\vb{x}}}} \, \bra{v_{j,{\vb{x}}}}.
\end{equation*}
The vectors $\ket{v_{j,{\vb{x}}}}$ are not orthogonal to each other nor normalized but they are all non-zero. We can obtain another valid frustration-free decomposition of $H$ by relabeling the ${\vb{x}}$ and from now on we will assume that ${\vb{x}}$ in $\ket{v_{j,{\vb{x}}}}$ is the lower-left point in the support of $\ket{v_{j,{\vb{x}}}}$.
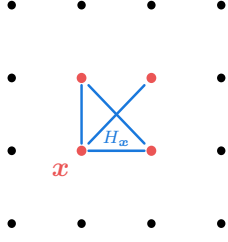
\begin{figure}[h]
    \centering
  \resizebox{3cm}{3cm}
  {
    \begin{tikzpicture}[line cap=round,line join=round]

    \definecolor{myred}{RGB}{235,85,85}
    \definecolor{myblue}{RGB}{25,120,220}

    \foreach \x/\y in {
        0/0,1/0,2/0,3/0,
        0/1,3/1,
        0/2,3/2,
        0/3,1/3,2/3,3/3
    }{
        \fill[black] (\x,\y) circle[radius=0.06];
    }

    \foreach \x/\y in {1/1,2/1,1/2,2/2}{
        \fill[myred] (\x,\y) circle[radius=0.07];
    }

    \draw[myblue, line width=1pt] (1.1,1) -- (1.9,1);
    \draw[myblue, line width=1pt] (1,1.1) -- (1,1.9);
    \draw[myblue, line width=1pt] (1.1,1.1) -- (1.9,1.9);
    \draw[myblue, line width=1pt] (1.1,1.9) -- (1.9,1.1);

    \node[myred, below left=1pt] at (1,1) {$\vb{x}$};
    \node[myblue, above=1pt] at (1.5,0.9) {\tiny $H_{\vb{x}}$};

    \end{tikzpicture}%
   }
    \caption{The term $H_{\vb{x}}$ describes the interaction between four (coarse-grained) sites, shown in blue lines. The subscript $\vb{x}$ refers to the lower left site of support of $H_{\vb{x}}$. The support of the $\ket{v_{j,{\vb{x}}}}$ is shown in red. }
    \label{fig:VSupport}
\end{figure}
Additionally, by coarse graining, i.e. grouping sites together into a site with more spin components, we can assume that the interaction range is one, see Figure~\ref{fig:VSupport}. In other words,
\begin{equation}
\label{eq:grain}
\ket{v_{j,{\vb{x}}}} = \ket{\delta_{\vb{x}}} a_{j,{\vb{x}}}  + \ket{\delta_{{\vb{x}} + e_1}} b_{j,{\vb{x}}} + \ket{\delta_{{\vb{x}} + e_2}} c_{j,{\vb{x}}} + \ket{\delta_{{\vb{x}} + e_1 + e_2}} d_{j,{\vb{x}}},
\end{equation}
where $a_{j,{\vb{x}}},b_{j,{\vb{x}}},c_{j,{\vb{x}}},d_{j,{\vb{x}}}$ are vectors in $\mathbb{C}^N$ indexed by $j,{\vb{x}}$.

\subsection{Edge Hamiltonian}
The edge system is defined on $\mathcal{H}_e = \ell^2(\mathbb{N} \times \mathbb{Z}, \mathbb{C}^N)$. The edge Hamiltonian $H_e$ is obtained from $H$ by dropping all terms $H_{\vb{x}}$ on the left hand side of the edge,
\begin{equation}
\label{eq:H_edge}
H_e = \sum_{{\vb{x}} \in \mathbb{N} \times \mathbb{Z}} H_{\vb{x}}.
\end{equation}
Notice that the edge Hamiltonian acts on $\mathcal{H}_e$ by our labeling convention. It satisfies the F$^4$ conditions but it is not necessarily gapped, as we will see in Section~\ref{sec:example}.
\subsection{Translation-invariant models}
 Let $T_{\vb{x}}$ be the standard representation of translations of $\mathbb{Z}^2$ on $\mathcal{H}$. The F$^4$ model is translation-invariant if
\begin{equation}
    \label{eq:H_as_sum_of_v}
 H=\sum_{\vb{x}\in \Zee^2}\sum_{j=1}^M T_{\vb{x}} \ket{v_j}\bra{v_j} T_{\vb{x}}^{-1},
\end{equation}
namely the vectors $a_{j,{\vb{x}}},b_{j,{\vb{x}}},c_{j,{\vb{x}}},d_{j,{\vb{x}}}$ in (\ref{eq:grain}) are $\vb{x}$ independent, and $\ket{v_j}=\ket{v_{(0,0),j}}$ are some vectors supported near the lattice origin. We immediately remark that in archetypal magnetic situations such as the Harper model, translation invariance after coarse-graining is possible only if the magnetic flux in a cell is rational. 

The Fourier transform 
$$
\mathcal{F} \, : \,\mathcal{H} \to L^2(\Tee^2, \tfrac{d\vb{k}}{4\pi^2};\mathbb{C}^N) \quad, \quad f \mapsto \widehat{f}(\vb{k}) = \sum_{\vb{x}\in \Zee^2} e^{i \vb{k} \vb{x}} f(\vb{x}).
$$
gives
\begin{equation}\label{eq:NN interaction}
\widehat{v}_{\vb{x},j}(\vb k) =  e^{i\vb{k} \vb{x}} \left( a_j + e^{ik_1} b_j + e^{i k_2} c_j + e^{i (k_1 + k_2)} d_j \right).
\end{equation}
The Hamiltonian $\widehat{H} = \mathcal{F} H \mathcal{F}^{-1}$ acts as an $N\times N$ matrix-valued function of $\vb{k}$, $\widehat{H\psi}\;(\vb{k}) = \widehat{H}(\vb{k}) \widehat{\psi}(\vb{k})$, with
\begin{multline}\label{eq: H of k}
\widehat{H}(\vb{k}) =  \sum_{j = 1}^M \widehat{v}_{0,j}(\vb{k}) \widehat{v}_{0,j}(\vb{k})^* 
=  \sum_{j=1}^M \Big(a_j a_j^* + b_j b_j^* + c_j c_j^* + d_j d_j^* +  \\ + e^{i k _1} (b_j a_j^*  +d_j c_j^*) + e^{-i k_1} (a_j b_j^* + c_j d_j^*) +  e^{i k_2}(c_j a_j^* + d_j b_j^*) + \dots\Big).
\end{multline}

We will not need the explicit expression for $\widehat{H}(\vb{k})$ but it will be crucial for us that, because of the finite range condition, $\widehat{H}(\vb{k})$ is a (matrix-valued) Laurent polynomial in $e^{ik_1}, e^{i k_2}$, which we shall call a trig-polynomial. 

We now turn to the fact that a translation-invariant F$^4$ Hamiltonian has a flat ground state band. 
\begin{defn}
\label{def:flatband}
We say that a translation-invariant Hamiltonian $G$ has a flat ground state band if for all $\vb{k}\in \Tee^2$,
\begin{enumerate}
\item $\ft{G}(\vb{k}) \geq 0$,
\item $\Ker(\ft{G}(\vb{k})) \setminus \{0\} \neq\emptyset$.
\end{enumerate}
\end{defn}

The connection between frustration-freeness and flat band is not new but we feel that it is worth stating explicitly.
\begin{prop}\label{prop:flat_band}
If $H$ is a translation-invariant F$^4$ Hamiltonian, then $H$ has a flat ground state band. Conversely, if $G$ is a translation-invariant finite range gapped Hamiltonian with a flat band, then $H=G^2$ is a gapped F$^4$ Hamiltonian with the same kernel.
\end{prop}

\begin{proof}
If $H$ is a translation-invariant F$^4$ Hamiltonian, then i. in Definition~\ref{def:flatband} follows immediately from~\eqref{eq: H of k}, and $\vb{k}\mapsto\widehat H(\vb{k})$ is a trig-polynomial. To show ii., we note that if $0\neq\Omega\in \Ker H$, then $\widehat H(\vb{k})\widehat\Omega(\vb{k}) = 0$ for almost all $\vb{k}\in\Tee^2$. If the set $Z=\{\vb{k}\in\Tee^2: \mathrm{det}(\widehat H(\vb{k}))=0\}$ had measure zero, then $\widehat\Omega(\vb{k})=0$ almost everywhere, and in turn $\Omega = 0$, which is a contradiction. Hence $\vert Z\vert>0$. Let $\mathrm{det}(\widehat H(\vb{k})) = \sum_{m,n} \beta_{m,n}e^{i(m k_1 + n k_2)}$. For any fixed $k_2$, let $Z_{k_2}=\{k_1\in\Tee: \mathrm{det}(\widehat  H(\vb{k}))=0\}$; the function $k_1\mapsto\mathrm{det}(\widehat H(k_1,k_2))$ is a Laurent polynomial in one variable, $\vert Z_{k_2}\vert >0$ if and only if the polynomial is the zero polynomial, namely $\mathrm{det}(\widehat H(k_1,k_2)) = 0$ for all $k_1$, since otherwise the set of zeros is finite. By Tonelli, $0<\vert Z\vert =\int_{0}^{2\pi}\vert Z_{k_2}\vert d k_2$, and this implies that $B = \{k_2: \mathrm{det}(\widehat H(\cdot,k_2)) = 0\}$ has positive measure. For any $k_2\in B$, all Fourier coefficients must vanish, namely $\alpha_m(k_2) = \sum_{n} \beta_{m,n}e^{i n k_2}=0$. But $\alpha_m(k_2)$ is a Laurent polynomial and it vanishes on a set of positive measure, hence it is uniformly zero. But this in turn implies that $\beta_{m,n} = 0$ for all $n,m$, namely $\mathrm{det}(\widehat H(\vb{k}))=0$ for all $\vb{k}$, which is ii. of Definition~\ref{def:flatband}.

For the converse statement, we first note that $\widehat{G^2}(\vb{k}) = \widehat{G}(\vb{k})\widehat{G}(\vb{k}) = \widehat{G}(\vb{k})\widehat{G}(\vb{k})^*$ since $\widehat{G}(\vb{k})$ is self-adjoint. In particular, $\widehat{G^2}(\vb{k})\geq 0$. Moreover, $G$ is finite range so that $\widehat{G}(\vb{k})$ is trig-polynomial and hence so is $\widehat{G^2}(\vb{k})$. This in turn implies that $H = G^2$ has finite range. Explicitly,
\begin{equation*}
\widehat{G^2}(\vb{k}) = \sum_{\alpha = 1}^N \ft u_\alpha(\vb{k})\ft u_\alpha(\vb{k})^*
\end{equation*}
where $\ft u_\alpha$ are the columns of $\ft G$ and therefore trig-polynomials. Hence $u_\alpha = \mathcal{F}^{-1}\ft u_\alpha$ are all compactly supported and $H = \sum_{\vb{x}} H_{\vb{x}}$ where $H_{\vb{x}} = \sum_\alpha T_{\vb{x}}\ket{u_\alpha}\bra{u_\alpha}T_{\vb{x}}^{-1}$. Now $\bigcap_{\vb{x}} \Ker(H_{\vb{x}}) = \Ker(G^2) = \Ker(G)$. With this, it remains to construct a non-vanishing vector $\widehat{\Omega}\in L^2(\Tee^2)$ such that $\widehat{G}\widehat{\Omega} = 0$. For this, we note that the uniform gap assumption for $\widehat{G^2}(\vb{k})$ and the continuity of $\vb{k}\mapsto\widehat{G^2}(\vb{k})$ implies that the Riesz projection $P(\vb{k})$ on $\Ker(\widehat{G}(\vb{k}))$ is continuous. Let $\vb{k}_0\in\Tee^2$ and let $\xi$ be a normalized vector in $\Ker(\widehat{G}(\vb{k_0}))$. Then for all $\vb{k}$ sufficiently close to $\vb{k}_0$, $\Vert P(\vb{k})\xi\Vert\geq \frac{1}{2}$, and so $\ft{\Omega}(\vb{k}) = \chi(\vb{k})P(\vb{k})\xi$ is a function in $L^2(\Tee^2)$ which satisfies $\widehat{G}(\vb{k})\ft{\Omega}(\vb{k}) = 0$ for all $k\in\Tee^2$, where $\chi$ is a bump function localized in a neighbourhood of $\vb{k}_0$. This concludes the proof. 
\end{proof}

\begin{rem}
(i) The central argument of the forward implication is simple: the determinant of the fibre Hamiltonians being a trig-polynomial because of the finite range condition, if it vanishes on a set of positive measure, then it must vanish everywhere. A slightly more sophisticated argument shows that the same conclusion holds in the more general case where the determinant is a real-analytic function on the torus.\\
\noindent (ii) Following up on the proof of the backward implication, if $U$ denotes a neighbourhood of $\vb{k}_0$ on which $\Vert P(\vb{k})\xi\Vert\geq \frac{1}{2}$, then replacing $\chi$ by any $g\in L^2(U)$ defines an injective map $\mathcal{F}^{-1}(g P(\cdot)\xi) = \int_{\Tee^2}g(\vb{k}) e^{- i\vb{k}\vb{x}} P(\vb{k})\xi \, d\vb{k}$ into $\Ker(G)$. This shows that the kernel, which corresponds to the flat band, is infinite dimensional. 
\end{rem}

\subsection{Bulk Index}
\label{BulkIndexSetUp}
We define the bulk index of a gapped F$^4$ Hamiltonian as the Chern number of the ground state bundle $\cup_{\vb{k} \in \Tee^2} \; \left( \{\vb{k}\} \times \Ker{\ft{H}(\vb{k})}\right)$.
\begin{thm}\label{thm:Bulk}
The bulk index of a gapped, translation-invariant F$^4$ Hamiltonian vanishes.
\end{thm}
\noindent As discussed above, this theorem is not new. Using $K$-theory, \cite{ChenTrivalFlatBand} proved that flat bands from local Hamiltonians are trivial. In Section~\ref{sec:bulk}, we give an alternative proof of this statement using an approach similar to~\cite{Dubail_2015}.

\subsection{Edge Index}
\label{subsec:edge_index}
The edge Hamiltonian~(\ref{eq:H_edge}) retains translation invariance in the $e_2$ direction. In this context, we use $k = k_2$ to denote the Fourier mode in the $e_2$ direction. Specifically, one performs partial Fourier transform,
$$
\mathcal{F}_2 : \mathcal{H}_e \to \ell^2(\mathbb{N}, \mathbb{C}^N)\otimes L^2(\mathbb{S}^1,\tfrac{dk}{2\pi}), \quad f \mapsto \widehat{f}(x_1,k) =  \sum_{x_2\in\Zee} e^{i k x_2} f(x_1, x_2),
$$
which gives 
$$
\widehat{v}_{{\vb{x}},j}(k) =  e^{i k x_2} \left( \delta_{x_1} a_j + \delta_{x_1 + 1} b_j + e^{i k} \delta_{x_1} c_j + e^{ik} \delta_{x_1 + 1} d_j \right).
$$
The Hamiltonian $\mathcal{F}_2 H_e \mathcal{F}_2^{-1}$ is then fibered over $k\in \mathbb{S}^1$ and each $\widehat H_e(k)$ is a nearest-neighbor hopping Hamiltonian on $\ell^2(\mathbb{N}, \mathbb{C}^N)$. Specifically,
\begin{equation}
\label{eq:edge_Ham_k}
\widehat H_e(k) = \sum_{x_1 \in \mathbb{N}}\sum_{j = 1}^M \; \widehat{v}_{(x_1, 0), j}(k)\; \widehat{v}_{(x_1, 0), j}(k)^*.
\end{equation}
The spectrum of $\widehat H_e(k)$ is manifestly non-negative and by the gap assumption, the spectrum of $\widehat H_e(k)$ in the interval $(0, \gamma)$ is discrete. Although it may in principle accumulate at $0$, we will show in Section~\ref{sec:edge} that it is in fact finite. The essential spectrum of $\widehat H_e(k)$ is equal to the essential spectrum of $\widehat H(k)$; and in particular $0$ belongs to the essential spectrum, see \cite{GrafPorta}. We define
\begin{equation}\label{eq:Gamma}
\Gamma = \{ (k, E) \in \mathbb{S}^1 \times (0, \gamma) : E \in \sigma(\widehat H_e(k)) \} 
\end{equation}
  By standard perturbation theory \cite{Kato}, for any point $(k,E) \in \Gamma$ there exists a neighborhood $U$ with base $k_- < k < k_+$ such that $\Gamma \cap U = \cup_{j} \{ (k, E_j(k)), k \in (k_-, k_+)\}$ for finitely many real analytic functions $E_j(k)$, which we shall refer to as local branches. Thus $\Gamma$ is locally a finite union of graphs of real analytic functions of $k$. We shall use this graph structure, which is stronger than analyticity of the set alone, and count branches with their spectral multiplicities.
  
 Globally, eigenvalue crossings may lead to a non-trivial permutation of the branches $E_j$ between $k=0$ and $k=2\pi$, thus preventing a single global analytic labeling of eigenvalues on $\mathbb{S}^1$; however, as we will explain, this plays no role in what follows, and the reader may keep in mind the generic situation illustrated in Figure~\ref{fig:branch cut}. 
\begin{figure}[t]
\centering
\definecolor{edgeblue}{RGB}{21,101,192}
\definecolor{bandgrey}{RGB}{78,78,84}
\definecolor{gammapurple}{RGB}{106,27,120}
\begin{tikzpicture}[
  x=1.3cm, y=3cm,
  band/.style     = {bandgrey, thick},
  edge/.style     = {edgeblue, thick},
  frame/.style    = {bandgrey!35, line width=0.5pt},
  zeroline/.style = {bandgrey!70, line width=0.6pt},
  drop/.style     = {gammapurple!60, line width=0.45pt,
                     dash pattern=on 1.6pt off 1.6pt},
]
\def\xmax{6.2832}   
\def\ymax{1.32}
\def\gam{0.19278}   
\def\Eedge{0.19278 - 0.17*sin(deg(2*\x+0.95)) - 0.060*sin(deg(\x+3.18))}
\def\Uone{0.70 + 0.028*cos(deg(\x+0.6))
              + 0.032*cos(deg(\x/2+0.20)) + 0.088*cos(deg(1.5*\x+0.60))}
\def\Utwo{0.70 + 0.028*cos(deg(\x+0.6))
              - 0.032*cos(deg(\x/2+0.20)) - 0.088*cos(deg(1.5*\x+0.60))}
\def\Cont{1.10 + 0.045*cos(deg(2*\x+1.0))}
\def\Barc{1.10 + 0.045*cos(deg(2*\x+1.0)) - 0.22*sin(deg(pi*(\x-1.05)/3.25))}

\path[pattern=north east lines, pattern color=bandgrey!55]
  plot[domain=0:\xmax, samples=200, smooth, variable=\x] ({\x},{\Cont})
  -- (\xmax,\ymax) -- (0,\ymax) -- cycle;
\draw[band] plot[domain=0:\xmax, samples=200, smooth, variable=\x] ({\x},{\Cont});

\draw[band] plot[domain=1.05:4.30, samples=160, smooth, variable=\x] ({\x},{\Barc});

\draw[band] plot[domain=0:\xmax, samples=200, smooth, variable=\x] ({\x},{\Uone});
\draw[band] plot[domain=0:\xmax, samples=200, smooth, variable=\x] ({\x},{\Utwo});

\foreach \k in {0.9465, 2.7310, 4.4113, 5.7192}
  \draw[drop] (\k,0) -- (\k,\gam);

\draw[gammapurple, line width=0.7pt] (0,\gam) -- (\xmax,\gam);
\node[gammapurple, left=2pt] at (0,\gam) {$\gamma$};

\begin{scope}
  \clip (0,\gam) rectangle (\xmax,\ymax);
  \draw[band] plot[domain=0:\xmax, samples=400, smooth, variable=\x] ({\x},{\Eedge});
\end{scope}
\begin{scope}
  \clip (0,0) rectangle (\xmax,\gam);
  \draw[edge] plot[domain=0:\xmax, samples=400, smooth, variable=\x] ({\x},{\Eedge});
\end{scope}
\node[edgeblue, font=\small] at (3.53,0.100) {$\Gamma$};

\foreach \k in {0.9465, 2.7310, 4.4113, 5.7192}
  \fill[gammapurple] (\k,\gam) circle (1.8pt);

\draw[frame] (0,0) rectangle (\xmax,\ymax);
\draw[zeroline] (0,0) -- (\xmax,0);
\node[bandgrey, left=2pt, font=\footnotesize] at (0,0) {$0$};
\foreach \k/\an/\lab in {0/north west/{$0$}, 0.9465/north/{$k_1$}, 2.7310/north/{$k_2$},
                         4.4113/north/{$k_3$}, 5.7192/north/{$k_4$},
                         6.2832/north east/{$2\pi$}} {
  \draw[frame] (\k,0) -- (\k,-0.018);
  \node[font=\footnotesize, anchor=\an, inner sep=2pt] at (\k,-0.018) {\lab};
}
\end{tikzpicture}
    \caption{The set $\Gamma$ is in blue. The eigenvalues of $\widehat H_e(k)$ below $\gamma$ can be seen as functions defined on the intervals $[k_2, k_3]$ and $[k_4, k_1]$. Analytical continuation around the circle may permute branches between $k=0$ and $k = 2 \pi$, so a global periodic analytic labeling need not exist. The argument does not require such a labeling.}
    \label{fig:branch cut}
\end{figure}
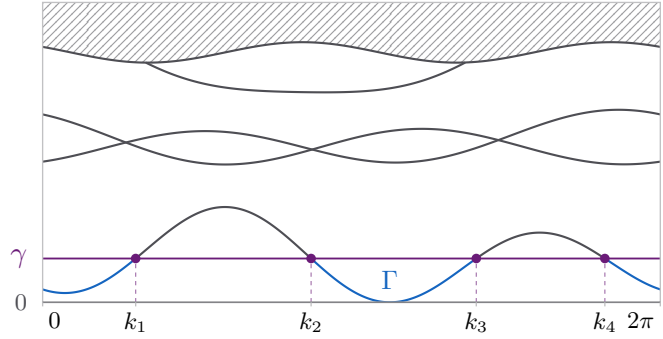

A point $(k,E) \in \Gamma$ is called regular if locally $\Gamma$ is the graph of a single analytic function $E(k)$ with $E'(k) \neq 0$.  An energy $\mu \in (0, \gamma)$ is called regular if the line $\mathbb{S}^1 \times \mu$ intersects $\Gamma$ only at regular points. All but finitely many energies in $(0, \gamma)$ are regular. 

\begin{defn}\label{def:Edge Index}
    The edge index of $H_e$ is the total number of signed crossings of $\Gamma$ across any regular energy, counted with spectral multiplicity, with increasing crossings counted positively.
\end{defn}
That this is well-defined, in particular, independent of the choice of regular energy $\mu$ follows from the observation that the edge index is the degree $\mathrm{deg}(f;\mu)$ of the map $f : \bar \Gamma \to \mathbb{R}$ defined by $(k,E) \mapsto E$, and the fact that $\mathrm{deg}(f;\mu)$ is independent of the choice of regular value, see Theorem~13.1.2 in~\cite{DubrovinNovikovFomenko1985}.\\

The edge index of $H_e$ also has an analytic expression, see for example~\cite{GrafIQHE}. It is given by
    \begin{equation}
    \label{eq:edge_index}
     -\int_{\mathbb{S}^1} \tr \left(g'(\widehat H_e(k)) \frac{d\widehat H_e(k)}{dk}\right)dk,
    \end{equation}
    for any $g\in C^1(\real)$, that is equal to $1$ on  $(-\infty,0]$, and equal to $0$ on $[\gamma,\infty)$.

The closure $\bar\Gamma\subseteq\mathbb S^1\times[0,\gamma]$ will be treated in the same branchwise manner. Lemma~\ref{BranchDef} below shows that its branches continue as nonnegative real analytic functions of $k$ through $E=0$. After including these continuations, separating intersecting branches and retaining multiplicities, we obtain a compact oriented curve $\mathcal C$. We take the upper cutoff $\gamma$ to be regular and write
\begin{equation}\label{eq:resolved-curve}
f:\mathcal C\longrightarrow\mathbb R,\qquad (k,E)\longmapsto E.
\end{equation}
The boundary $\partial\mathcal C$ consists of endpoints of the continued branches, not the topological boundary of $\bar\Gamma$ in the cylinder. Gluing analytic pieces at crossings or at $k=0=2\pi$ creates no boundary; a permutation of branch labels is allowed. 

The figure illustrates the difference between the topological boundary and the end-points of branches. 

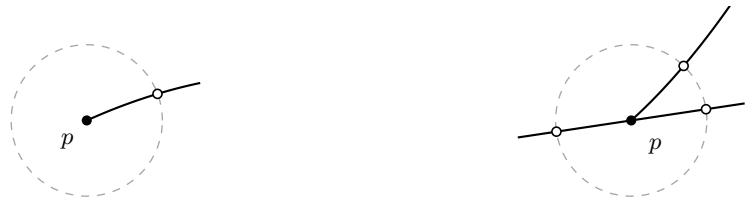
\begin{figure}[h]
\centering
\begin{tikzpicture}[
  x=1cm, y=1cm,
  curve/.style   = {thick, black},
  link/.style    = {draw=black, fill=white, line width=0.6pt},
  ball/.style    = {dashed, gray!70, line width=0.5pt},
  plabel/.style  = {font=\footnotesize},
  nlabel/.style  = {font=\footnotesize, anchor=north},
]

\begin{scope}
  \clip (-1.5,-1.5) rectangle (1.5,1.5);
  \draw[ball] (0,0) circle (1);
  \draw[curve, domain=0:1.5, smooth, variable=\x] plot ({\x},{0.45*\x-0.08*\x*\x});
\end{scope}
\fill (0,0) circle (1.9pt);
\node[plabel, anchor=north east] at (-0.04,-0.04) {$p$};
\draw[link] (0.935,0.351) circle (1.7pt);

\begin{scope}[shift={(7.2,0)}]
  \begin{scope}
    \clip (-1.5,-1.5) rectangle (1.5,1.5);
    \draw[ball] (0,0) circle (1);
    \draw[curve, domain=-1.5:1.5, smooth, variable=\x] plot ({\x},{0.15*\x});
    \draw[curve, domain=0:1.32, smooth, variable=\x] plot ({\x},{0.9*\x+0.2*\x*\x});
  \end{scope}
  \fill (0,0) circle (1.9pt);
  \node[plabel, anchor=north west] at (0.10,-0.10) {$p$};
  \draw[link] (0.989,0.148) circle (1.7pt);
  \draw[link] (-0.989,-0.148) circle (1.7pt);
  \draw[link] (0.693,0.720) circle (1.7pt);
\end{scope}

\end{tikzpicture}
\caption{Both pictures show boundary points of branches though only the one on right is a topological boundary point.}
\label{fig:branching}
\end{figure}

Local perturbation theory already rules out branch endpoints in $(0,\gamma)$. The two remaining possibilities, $0$ and $\gamma$, are different: $\gamma$ is an arbitrary cutoff, through which the branches continue analytically, whereas $0$ belongs to the essential spectrum. At zero, perturbation theory does not apply directly and a priori a branch might terminate. The key technical result is that every such branch continues through zero as an analytic graph over $k$.

\begin{lemma}
\label{BranchDef}
Let $H$ be a gapped, translation invariant F$^4$ Hamiltonian. At every point of $\bar\Gamma$ with $E<\gamma$, including $E=0$, the set is locally a finite union of graphs of nonnegative real analytic functions of $k$, defined on a full neighborhood of that momentum. The continuations of positive-energy branches preserve their spectral multiplicities. In particular, $f(\partial\mathcal C)\subseteq\{\gamma\}$.
\end{lemma}
The proof of the lemma is in Section~\ref{sec:edge}. It immediately follows that
\begin{thm}
\label{thm:edge}
The edge index of a gapped, translation invariant F$^4$-Hamiltonian is 0. 
\end{thm}

\begin{proof}
Let $\gamma$ be in the gap of $\sigma(H)$ such that all energies in $(0, \gamma)$ are regular. The degree $(0,\gamma)\ni \mu\mapsto \mathrm{deg}(f;\mu)$ is constant in any connected component of $\mathbb{R} \setminus f(\partial \bar{\Gamma})$. The claim of the theorem follows immediately from this and Lemma~\ref{BranchDef} since $\mathrm{deg}(f; \mu) = 0$ for $\mu <0$ and $0\notin f(\partial \bar{\Gamma})$. 
\end{proof}

It might be worth repeating the same reasoning analytically.  By Lemma~\ref{BranchDef} and its proof, branches that touch zero continue as real analytic functions of $k$. Cutting into finitely many analytic pieces, including at the seam $k=0=2\pi$, gives functions $E_l:D_l\to[0,\gamma]$ on closed intervals $D_l=[a_l,b_l]\subseteq[0,2\pi]$, counted with multiplicity, such that
$$
\bar{\Gamma} = \bigcup_l \bigcup_{k\in D_l} \big\{ (k, E_l(k)) \big\}.
$$
The fundamental theorem of calculus applied to the expression~(\ref{eq:edge_index}) of the edge index yields that the index is equal to
$$\sum_l\bigl[g(E_l(a_l))-g(E_l(b_l))\bigr].$$
At artificial cuts the endpoint contributions cancel, with multiplicity, even if analytic labels are permuted. By Lemma~\ref{BranchDef}, a genuine branch endpoint can only have energy $\gamma$, where $g(\gamma)=0$. Thus the edge index vanishes.
\begin{rem}
If the Hamiltonian has time-reversal symmetry, then the Hall conductance is zero but there exists a $\mathbb{Z}_2$ index, see e.g.~\cite{GrafPorta} for details on the setting and definitions. Informally, the edge $\mathbb{Z}_2$ index counts number of crossings modulo $2$. The proof we have just presented yields in that case that the $\mathbb{Z}_2$ index vanishes for gapped, translation-invariant F$^4$-Hamiltonians: Indeed, each branch that touches zero necessarily contributes even number of crossings.
\end{rem}

\begin{rem}
The relationship between locality and the Hall conductance has a rich history initiated by~\cite{arovas1988} in the literature, in particular in the free Fermion case. The statement of Proposition~\ref{prop:flat_band} is that, under a finite range assumption, a flat band at zero energy and a frustration-free ground state space are mathematically equivalent. Of course, this does not exclude that excited bands have non-trivial dispersion relations. 

Then our main result, Theorem~\ref{thm:edge}, states that a fully filled flat band has zero Hall conductance. The finite range assumption is essential here and not a mere technicality. If $P$ is \emph{any} translation-invariant Fermi projection with Fermi level in a spectral gap, in particular it may have a non-zero Hall conductance, then the `flattened Hamiltonian' $(1-P)$ has exponentially decaying matrix elements and it is frustration-free: This can be  seen in the second quantized formalism since the vacuum is a ground state for $\sum_{\vb{x},\vb{y}}(1-P)_{\vb{x},\vb{y}}a\str_{\vb{x}}a_{\vb{y}}$. A family of similar examples are Haldane's pseudo-potential Hamiltonians~\cite{HaldanePseudo} on thin cylinders, which have Laughlin's wavefunctions~\cite{Laughlin1983} as their frustration-free and gapped ground states, but with Gaussian-decaying rather than finite-range interactions and the Hall conductance is non-zero, see e.g.~\cite{FQHE}. 

A frustration free Hamiltonian with exponentially decaying tailes can be truncated with arbitrary precision, $\epsilon$, to a frustration free Hamiltonian with compactly supported terms. In such case, a new band with distance $\epsilon$ from the ground state band might appear. If the ground state band had non-zero Chern number then this new band have to appear and will carry the same Chern number leaving the ground state band with zero Chern number.

One may also contrast our result with the case of (translation-invariant) commuting projector Hamiltonians. Those are also known to have vanishing Hall conductance, first in the finite range case~\cite{kapustin2020local}, then in the case of superpolynomial decay~\cite{zhang2022vanishing}. For free fermions with Fermi projection $P$, this question is intimately related to the existence of a localized Wannier basis. Indeed, if there is an exponentially localized orthonormal basis $\{f_j\}_{j\in\bbN}$ of the range of $P$, then again $\sum_{j\in\bbN}a\str(f_j)a(f_j)$ is a frustration-free commuting projector Hamiltonian which has vanishing Hall conductance. In this case, relaxing the locality assumption allows for a transition to a non-trivial topological state: the localization-topology correspondence~\cite{monaco2018optimal} shows a sharp topological transition as a function of the polynomial decay rate of the Wannier functions. 
\end{rem}

\begin{rem}
Finally, let us emphasize the robustness of the edge approach that we are advertising in this note. While the bulk proof, which relies on the syzygy theorem requires full translation invariance, the edge approach does not at all. As we shall see in the next section, the partial Fourier transform $\mathcal{F}_2$ is used only to produced a real analytic family $k \mapsto \ft{H}_e(k)$ of half-line operators indexed by a circle. Everything else uses only Assumption~\ref{AssumptionF4} and the gap. The same structure can be obtained by defining the model on a cylinder using the identification $\ell^2(\bbN \times (\Zee/L\Zee), \Cee^N) \cong \ell^2(\bbN, \Cee^{NL})$, with the circle of half-line operators being produced by threading a flux, which is done by adding a phase $\exp(\pm i \varphi/L)$ to all hopping terms along the angular direction. This requires no translation invariance whatsoever. The outcome is that a gapped F$^4$ Hamiltonian on a cylinder has vanishing edge index, with no assumption of translation invariance, covariance or ergodicity; in particular the statement of Theorem~\ref{thm:edge} applies to each realization of a disordered model. In fact, the Hamiltonians may well be interacting. Now, for the arguments to hold verbatim, the spectrum of the flux Hamiltonians $\{H(\varphi):\varphi\in[0,2\pi]\}$ must be assumed to be such that $\sigma ( H(\varphi) ) \cap (0,\gamma) = \emptyset$ for all $\varphi$, in which case they all satisfy the F$^4$ condition if $H(0)$ does.

\end{rem}


\section{Analysis of Edge Modes}
\label{sec:edge}

We now turn to the technical heart of the paper, namely the proof of Lemma \ref{BranchDef}. We introduce an auxiliary Hamiltonian $H_a$, which is a bulk operator and can therefore be compared with $H$, but it corresponds to a system where two half-spaces have been disconnected, allowing for a relation with the edge Hamiltonian $H_e$.

\subsection{Slicing}
\label{subsec:edgesetup}
Let $H_a$ be the operator on $\ell^2(\mathbb{Z}^2,\Cee^N)$ defined by
\begin{equation}
    H_a = \sum_{\vb{x} \in \mathbb{Z}^2, x_1 \neq 0} H_{\vb{x}}.
\end{equation}

With our choice of labeling of the interaction terms $H_{\vb{x}}$, the subspace $\ell^2(\mathbb{N} \times \mathbb{Z}, \Cee^N)$ is an invariant subspace of $H_a$ and the restriction of $H_a$ to this subspace is the edge Hamiltonian $H_e$.

If $H$ is translation-invariant, then $H_a$ retains invariance in the vertical direction. As for the edge Hamiltonian, the partial Fourier transform $\mathcal{F}_2$ yields the family of operators on $\ell^2(\Zee,\Cee^N)$
\begin{equation}
    \label{eq:aux_H_frame_vec}
    \widehat H_a(k) = \sum_{x_1 \neq 0,\, j}  \widehat{v}_{(x_1, 0), j}(k) \,  \widehat{v}_{(x_1, 0), j}(k)^*
\end{equation}
parametrized by $k\in \mathbb{S}^1$. For every $k\in \mathbb{S}^1$, this is a finite rank perturbation of the partially Fourier transformed bulk Hamiltonian
\begin{equation}
    \widehat H(k) = \sum_{x_1\in \Zee,\, j}  \widehat{v}_{(x_1, 0), j}(k) \,  \widehat{v}_{(x_1, 0), j}(k)^*.
\end{equation}
In fact, $\widehat H_a(k)\leq \widehat H(k)$ since the difference $\sum_{j=1}^M  \widehat{v}_{(0, 0), j}(k) \,  \widehat{v}_{(0, 0), j}(k)^*$ is non-negative and has rank at most $M$.
It follows that the spectrum of $\widehat  H_a(k)$ has at most finitely many eigenvalues in the gap $(0, \gamma)$. Again by analytic perturbation theory,
 $$\Gamma_a = \{ (k, E_a) \in \mathbb{S}^1 \times (0, \gamma) : E_a \in \sigma(\ft{H}_a(k)) \}, $$ 
is locally a finite union of graphs of real analytic functions of $k$. As in~\eqref{eq:resolved-curve}, let $\mathcal C_a$ denote the resolved branches, including the continuations through zero supplied below, and let $f_a:\mathcal C_a\to\mathbb R$ be the height map. We have the following analog of Lemma~\ref{BranchDef}:
\begin{lemma}\label{lem:Ha}
Let $H$ be a gapped, translation-invariant F$^4$ Hamiltonian and let $\Gamma_a$ be as above. At every point with $E<\gamma$, including $E=0$, $\bar\Gamma_a$ is locally a finite union of graphs of nonnegative real analytic functions of $k$ on a full neighborhood of that momentum. Positive-energy branches retain their multiplicities through zero. In particular, $f_a(\partial\mathcal C_a)\subseteq\{\gamma\}$.
\end{lemma} 
\begin{proof}
Consider the projection $P(k)$ onto the kernel of the bulk Hamiltonian $\ft{H}(k)$. By assumption, $0$ is an isolated point in the spectrum of the real-analytic $\ft{H}(k)$ for all $k\in\mathbb{S}^1$, and so  $k\mapsto P(k)$ is real-analytic. The unitary family $U(k)$ solving $\dot U(k) = [\dot P(k),P(k)]U(k)$ with $U(0) = \mathbb{I}$ is such that $U(k)^* P(k) U(k) = P(0)$ for all $k\in [0,2\pi]$, and it is also real-analytic. It may not be $2\pi$-periodic, and this is irrelevant for the argument. The family of Hamiltonians $\tilde{H}(k) = U(k)^* \ft{H}(k) U(k)$ annihilates $\mathrm{Ran}(P(0))$, and by the gap assumption
\begin{equation*}
\sigma\Big(\tilde{H}(k)\vert_{1 - P(0)}\Big)\subseteq[\gamma,\infty).
\end{equation*}
Recalling that $H_a$ is obtained from $H$ by removing some interaction terms, the frustration-free condition implies that $\Ker(H_a)\supseteq\Ker(H)$ and so $\ft{H}_a(k) P(k) = 0$. This implies that $\mathrm{Ran}(P(0))\subseteq\Ker(\tilde{H}_a(k))$, where
$$
\tilde{H}_a(k) = U(k)^* \ft{H}_a(k) U(k)
$$
for all $k\in[0,2\pi]$. In turn, the eigenvalues of $\ft{H}_a(k)$ in $(0,\gamma)$ are in the spectrum of $\tilde{H}_a(k)\vert_{1-P(0)}$. \par
The non-negative operator $\tilde{H}_a(k)| _{1-P(0)}$ is a finite rank perturbation of $\tilde{H}(k)| _{1-P(0)}$, and the latter has no spectrum in $(-\infty,\gamma)$.
It follows from perturbation theory that for each $k\in \mathbb{S}^1$, the spectrum of $\tilde{H}_a(k)| _{1-P(0)}$ has at most a finite number of eigenvalues in $[0, \gamma)$, and none in $(-\infty,0)$. As already argued before, perturbation theory implies that these eigenvalues are locally analytic functions, which we denote $\tilde E_{a,j}(k)$. For any $k\in\mathbb{S}^1$, the spectrum of $\tilde{H}_a(k)| _{1-P(0)}$ below $\gamma$ is the union of $\tilde E_{a,j}(k)$. It follows that the set
$$
\tilde{\Gamma}_a = \{ (k, \tilde E_a) \in \mathbb{S}^1 \times (-\infty, \gamma) : \tilde E_a \in \sigma(\tilde{H}_a(k)| _{1-P(0)}) \}
$$
is locally a finite union of graphs of real analytic functions of $k$, including at zero. Its branches can terminate only at the cutoff $E=\gamma$, since $\tilde H_a(k)|_{1-P(0)}$ has no essential spectrum below $\gamma$.

Now, the set $\Gamma_a$ is a subset of $\tilde{\Gamma}_a$ that consists of all points $(k, E)$ in $\tilde{\Gamma}_a$ for which $E>0$. Since $\tilde{\Gamma}_a$ is closed in $\mathbb{S}^1 \times (-\infty, \gamma)$, the closure $\bar{\Gamma}_a\cap (\mathbb{S}^1 \times (-\infty, \gamma))$ is also a subset of $\tilde{\Gamma}_a$. We now analyze the points in $\tilde{\Gamma}_a$ such that $E=0$ (which may or may not belong to $\bar\Gamma_a$). If $(k_*,0)\in \tilde{\Gamma}_a$, then locally, there are $\tilde E_{a,j}(k)$ such that $\tilde E_{a,j}(k_*)=0$. Recalling that $\tilde{H}_a(k)| _{1-P(0)}$ is non-negative, analyticity implies that either $\tilde E_{a,j}(k)>0$ for all $k\neq k_*$ in a neighborhood of $k_*$ or $\tilde E_{a,j}(k)$ vanishes identically. In the former case, the branch belongs to $\bar\Gamma_a$ on both sides of $k_*$ and extends through $(k_*,0)$ as a nonnegative real analytic graph. Identically-zero branches do not themselves contribute to the closure of the positive spectrum. Discarding those branches therefore leaves, near every point of $\bar\Gamma_a$ at zero energy, a finite union of full analytic graphs with their continued multiplicities. No such branch terminates there, and hence $f_a(\partial\mathcal C_a)\subseteq\{\gamma\}$.
\end{proof}
A consequence of the lemma and of its proof is that if we pick $\gamma$ small enough so that all energies in $(0, \gamma]$ are regular, then there exists a finite family of analytic functions $E_{a,l}:D_{a,l}\to [0,\gamma]$ defined on closed intervals $D_{a,l} \subseteq \mathbb{S}^1$, possibly $D_l = \mathbb{S}^1$, such that
$$
\bar{\Gamma}_a = \bigcup_l \bigcup_{k \in D_{a,l}} \big\{(k, E_{a,l}(k)) \big\}.
$$
The corresponding continued finite-rank eigenprojections $P_{a,l}(k)$, obtained from the reduced problem and conjugated back by $U(k)$, are analytic. At crossings and zero contacts they need not be full spectral projections of $\ft H_a(k)$. As pointed out in the proof, $\ft{H}_a(k)$ has, in general, additional `zero modes' $E(k)=0$.  

\subsection{Proof of Lemma \ref{BranchDef} }

We decompose $\ell^2(\mathbb{Z}, \mathbb{C}^N) = \ell^2(\mathbb{N}, \Cee^N) \oplus \ell^2(\mathbb{Z} \setminus \mathbb{N}, \mathbb{C}^N)$. By construction,
$$
\ft{H}_a(k) = \ft{H}_e(k) \oplus L(k),
$$
from which we immediately conclude that $\Gamma_a = \Gamma\cup \Gamma_L$, and so
$$
\bar{\Gamma}_a = \bar{\Gamma} \cup \bar{\Gamma}_L.
$$
Here, $\Gamma$ was defined in~\eqref{eq:Gamma}, and $L$ and $\Gamma_L$ are obtained from $H$ as are $H_e$ and $\Gamma$. Note that this equality does not give $\partial\bar\Gamma\subseteq\partial\bar\Gamma_a$: a branch of $\Gamma_L$ through a point of $\partial\bar\Gamma$ would change the parity of its branching number in $\bar\Gamma_a$. This is why the proof below identifies branches rather than endpoint sets.

Let $Q$ be the multiplication operator in $\ell^2(\mathbb{Z}, \mathbb{C}^N)$ by the characteristic function of the right half-space. The operator
\begin{equation*}
\bar H_e(k) = \ft{H}_a(k) Q = Q \ft{H}_a(k) Q
\end{equation*}
is the extension of $\ft{H}_e(k)$ by $0$ from $\ell^2(\bbN,\mathbb{C}^N)$ to $\ell^2(\bbZ,\mathbb{C}^N)$. Since $\ft{H}_a(k)$ commutes with $Q$, so do its spectral projections and in particular, for all $l$, $\bar P_l(k) = QP_{a,l}(k)$ is a projection, possibly zero. If it is nonzero, then it is a spectral projection of $\bar H_e(k)$, and any vector in its range is supported on the right half-line. It follows that $\bar P_l(k)= P_l(k)$ up to a relabeling of the branches which can be made consistently by continuity, and any spectral projector of $\ft{H}_e(k)$ is obtained in this way. It follows that $E_l(k)=E_{a,l}(k)$ for any branch of eigenvalues $E_l(k)$ of $\ft{H}_e(k)$. Since the branches $k\mapsto E_{a,l}(k)$ terminate only at energy $\gamma$ by Lemma~\ref{lem:Ha}, so do $k\mapsto E_l(k)$. \teop

Finally, we note that the proof shows that, as already used in Section~\ref{subsec:edge_index}, the number of branches of $\widehat H_e$ in $(0,\gamma)$ is finite.

\section{Example with gapless edge states}
\label{sec:example}
In this section, we construct an $F^4$ Hamiltonian, which has vanishing index by the previous sections, but has nevertheless a gapless edge spectrum.

We pick $N=2$, $M=2$, and let $e=\begin{pmatrix}1 \\0\end{pmatrix}\in\bbC^2$. Let $X = e^{ik_1}, Y = e^{ik_2}$ and
\begin{equation}\label{eq:ex}
  \ft{v}_1(\vb{k}) = \bigl(1+2X\bigr)\,e,
  \quad
  \ft{v}_2(\vb{k}) = \bigl(1+3X+XY\bigr)\,e,
\end{equation}
In real space this is $H=\sum_{\vb{x}\in\bbZ^2}H_{\vb{x}}$, $H_{\vb{x}}=\sum_{j=1,2}|v_{j,\vb{x}}\rangle\langle v_{j,\vb{x}}|$
with $v_{1,\vb{x}}= \delta_{\vb{x}}\otimes e + 2\,\delta_{\vb{x}+e_1}\otimes e$ and
$v_{2,\vb{x}}= \delta_{\vb{x}}\otimes e+3\,\delta_{\vb{x}+e_1}\otimes e+\delta_{\vb{x}+e_1+e_2}\otimes e$. The model is frustration-free since $\ket{v_{1}},\ket{v_2}$ are both colinear with $e$ and so $\Ker H \supseteq l^2(\bbZ^2)\otimes e^\perp$.

In Fourier space, $\ft{H}(\vb{k})=\left(|f_1(k_1)|^2+|f_2(\vb{k})|^2\right)\,ee^*$ with
\begin{equation*}
   f_1(X) = 1+2X, \qquad f_2(X,Y)= 1+(3+Y)X,
\end{equation*}
and the problem is reduced to the analysis of the two polynomials $f_1,f_2$. Now, $\ft{H}(\vb{k})$ has rank $1$ for all $\vb{k}$ with eigenvalue branch 
\begin{equation*}
\lambda(\vb{k}) = |f_1(k_1)|^2+|f_2(\vb{k})|^2 = 16 + 10\cos(k_1) + 6\cos(k_2) + 2\cos(k_1 + k_2).
\end{equation*}
Hence the bulk gap is $2$, reached at $\vb{k} = (\pi,\pi)$.

\subsection*{Edge spectrum}
Restrict to the half plane $x_1 > 0$ and Fourier transform in the vertical direction. The $e$-component of the resulting operator is the Jacobi matrix on $\ell^2(\bbN)$
\begin{equation}\label{eq:jacobi}
  \ft{H}_e(k_2) =
  \begin{pmatrix}
    2      & \bar R   &        &  \\
    R      & 2+Q      & \bar R &  \\
           & R        & 2+Q    & \ddots\\
           &          & \ddots & \ddots
  \end{pmatrix},
\end{equation}
where $Q = 14 + 6\cos(k_2)$ and $R = 5 + e^{ik_2}$. It has a constant `potential' $2+Q$, constant hoppings $R$ and a boundary defect $-Q$. In order to find the spectrum, we conjugate by the diagonal matrix $R^n/\vert R\vert^n\delta_{n,m}$ to get the unitarily equivalent matrix with $R,\bar R$ replaced by $r=\vert R\vert$, in particular purely real. For $n\geq2$, the recursion relation $r\psi_{n-1} + q\psi_n + r\psi_{n+1}=E\psi_n$ (here $q = 2+Q$), is solved by
\begin{equation}\label{eq:EV half space}
\psi_n = z^n,\qquad E = q + r(z + z^{-1})
\end{equation}
The solution is non-trivial and square summable if and only if $0<\vert z\vert<1$ which is mapped bijectively under the Joukowski map to $\mathbb{C}\setminus[-2,2]$. The boundary condition reads $2\psi_1 + r\psi_2 = E\psi_1$, namely
\begin{equation*}
z = \frac{r}{2-q} = -\frac{\vert R\vert}{Q}
\end{equation*}
The unique bound state exists if and only if $Q>\vert R\vert$, which is the case for all $k_2$, in which case
\begin{equation}\label{eq:branch}
  E(k_2) = 2-\frac{|R|^2}{Q}=\frac{2+2\cos k_2}{14+6\cos k_2}.
\end{equation}
We immediately check that $0\leq E(k_2)\leq \frac{1}{5}<2$, and that it touches the flat band at $k_2 = \pi$: The edge is gapless, while the bulk gap is $\gamma=2$. Note that the bulk band can be recovered by picking $\vert z\vert = 1$ in~\eqref{eq:EV half space}, and recalling that the unit circle is mapped to the interval $[-2,2]$ to obtain the interval $[2+Q-2\vert R\vert, 2+Q+2\vert R\vert]$.

The zero mode at $k_2 = \pi$ can be computed directly from~\eqref{eq:jacobi}, and is such that $\psi_{n} = (-1/2)^{n-1}\psi_1$ for all $n\in\mathbb{N}$ --- within the $e$-component. It is indeed at $0$ energy and exponentially localized at the edge, namely a gapless edge mode, but carries no current along it and the edge index vanishes. Recalling the definition of the boundary $\partial\mathcal C$, this branch does not terminate or arise at $k_2 = \pi$, and so $\bar\Gamma$ is a single closed loop below $\gamma$, see Figure~\ref{fig:touching}.

\begin{figure}[h]
\begin{center}
\begin{tikzpicture}[scale=0.75]
\begin{scope}
  \def\bandlo{16+6*cos(deg(\x)) - 2*sqrt(26+10*cos(deg(\x)))}
  \def\bandhi{16+6*cos(deg(\x)) + 2*sqrt(26+10*cos(deg(\x)))}
  \def\edge{(2+2*cos(deg(\x)))/(14+6*cos(deg(\x)))}
  \begin{scope}[xscale=6/6.2832, yscale=5/12]
    \clip (0,0) rectangle (6.2832,12);
    \fill[bandgrey!22]
      plot[domain=0:6.2832,samples=200] (\x,{\bandlo}) --
      (6.2832,12) -- (0,12) -- cycle;
    \draw[bandgrey,very thick,domain=0:6.2832,samples=200] plot (\x,{\bandlo});
    \draw[bandgrey,very thick,domain=0:6.2832,samples=200] plot (\x,{\bandhi});
    \draw[gammapurple,thick] (0,2) -- (6.2832,2);
    \draw[edgeblue,very thick,domain=0:6.2832,samples=300] plot (\x,{\edge});
  \end{scope}
  \draw[->] (0,0) -- (6.5,0) node[right] {$k_2$};
  \draw[->] (0,0) -- (0,5.4) node[above] {$E$};
  \foreach \p/\l in {0/$0$, 3/$\pi$, 6/$2\pi$} \draw (\p,0) -- (\p,-0.08) node[below]{\l};
  \draw (0,{5*2/12}) -- (-0.08,{5*2/12}) node[left]{\textcolor{gammapurple}{$\gamma=2$}};
\end{scope}
\begin{scope}[xshift=8.6cm]
  \def\edge{(2+2*cos(deg(\x)))/(14+6*cos(deg(\x)))}
  \begin{scope}[xscale=6/6.2832, yscale=5/0.25]
    \clip (0,0) rectangle (6.2832,0.25);
    \draw[edgeblue,very thick,domain=0:6.2832,samples=400] plot (\x,{\edge});
  \end{scope}
  \draw[->] (0,0) -- (6.5,0) node[right] {$k_2$};
  \draw[->] (0,0) -- (0,5.4) node[above] {$E$};
  \foreach \p/\l in {0/$0$, 3/$\pi$, 6/$2\pi$} \draw (\p,0) -- (\p,-0.08) node[below]{\l};
  \draw (0,{5*0.2/0.25}) -- (-0.08,{5*0.2/0.25}) node[left]{\small $\tfrac15$};
  \draw[dashed,bandgrey] (3,0) -- (3,5.2);
  \node[edgeblue,anchor=west] at (4.3,{5*0.10/0.25}) {\small $E(k_2)$};
\end{scope}
\end{tikzpicture}
\end{center}
\caption{The bulk band (grey), the gap threshold $\gamma=2$ (purple) and the edge branch \eqref{eq:branch} (blue) on the left. The latter is scaled on the right. The edge branch touches $E=0$ tangentially at $k_2=\pi$.
}
\label{fig:touching}
\end{figure}
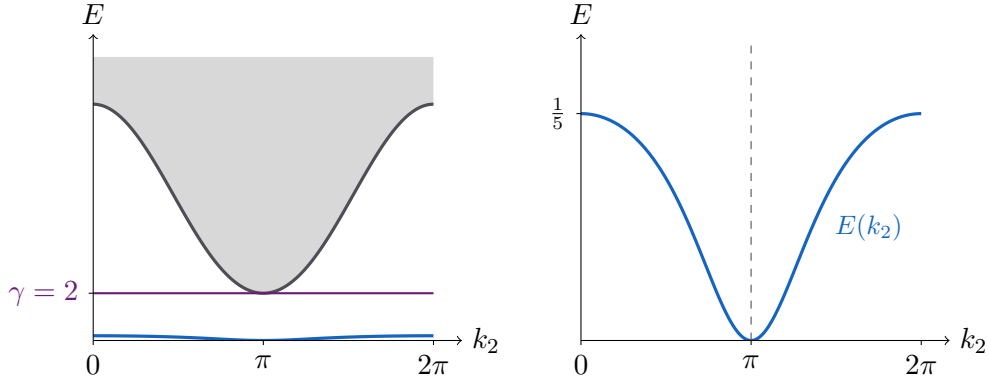
 
\section{Bulk Index}\label{sec:bulk}
We turn to the proof of Theorem~\ref{thm:Bulk}. We define the excited band bundle to be the subset of $\Tee^2 \times \Cee^N$
    \begin{equation}
      E=\bigsqcup_{\vb{k}\in \Tee^2} \mathrm{Im} \big(\ft{H}(\vb{k})\big) 
    \end{equation}
    with the subset topology. Similarly the ground state bundle is
        \begin{equation}
      G=\bigsqcup_{\vb{k}\in \Tee^2} \mathrm{Ker}\big(\ft{H}(\vb{k})\big) 
    \end{equation}
The bulk index is the first Chern number of bundle $G$ and it is equal to the Hall conductance, see~\cite{TKNN,A.S.S.ChernNumber}.

The key input in the proof is that by the gap condition, the dimension of kernel and of the image of $\widehat{H}(\vb{k})$ are constant over $\vb{k} \in \Tee^2$. So let
\begin{equation} \label{ConstRank}
M=\mathrm{Rank}_\Cee\left(\ft{H}(\vb{k})\right),
\end{equation}
namely, E is a rank $M$ complex vector bundle.

Since $\mathbb{C}^N = \mathrm{Ker}(\ft{H}(\vb{k})) \oplus \mathrm{Im}(\ft{H}(\vb{k}))$, the bundle $G\oplus E$ is the trivial bundle, in particular its Chern number vanishes. Since Chern number is additive, we can equivalently study $G$ or $E$ and it will be more convenient to  focus on $E$. Rather than studying $E$ directly, we will consider $\ft{H}(\vb{k})$ as the evaluation of a matrix with entries in the ring of Laurent polynomials. Let 
$$X=e^{ik_1},\; Y=e^{ik_2} \text{ and } R=\Cee\left[X^{\pm1},Y^{\pm1}\right].$$
Seeing $\widehat{H}$ is $N \times N$ matrix over $R$, we define
$$\mathrm{Im}_R(\widehat{H})=\left\{\widehat{H} \psi:\psi \in R^N\right\},$$ 
which is a finitely generated $R-$module. 

With this, the proof of Theorem \ref{thm:Bulk} will develop as follows: First, by Hilbert's syzygy theorem, there is an exact sequence of $R-$modules ending at the module $\mathrm{Im}_R(\widehat{H})$; Next, when evaluating the polynomials at $X=e^{ik_1},Y=e^{ik_2}$, the exact sequence of $R-$module descends to a family, parametrized by $\vb{k}\in\Tee^2$, of exact sequences of $\Cee$-vector spaces; Finally, we use the exactness to show that $E$ is a successive quotient of trivial bundles from trivial bundle.

\subsection{Proof of Theorem \ref{thm:Bulk}}

By the syzygy theorem \cite{rotman}, the finitely generated $R-$module $\mathrm{Im}_R(\widehat{H})$ over the ring $R$ has a free resolution:  
\begin{equation}
0 \xrightarrow{h_{3}} R^{N_2} \xrightarrow{h_{2}} R^{N_1} \xrightarrow{h_{1}} R^N \xrightarrow{h_0=\ft{H}} \mathrm{Im}_R(\widehat{H}) \to 0,\label{Syzygy}
\end{equation}
namely, $\Ker_R(h_j)=\mathrm{Im}_R(h_{j+1})$ and $h_j$ are $R$-module homomorphisms, namely matrices with entries in~$R$.

\begin{rem}[Meaning of the exact sequence]\label{rem:sequence-meaning}
Under Fourier transform, an element of $R^N$ corresponds to a compactly supported one-particle wavefunction, and multiplication by a monomial $X^mY^n$ translates that wavefunction. The columns $a_\alpha=\widehat H e_\alpha$ are therefore the Fourier transforms of $H\delta_{0,\alpha}$. The module $\operatorname{Im}_R(\widehat H)$ consists of finite linear combinations of their translates, or equivalently of vectors $H\psi$ with $\psi$ compactly supported.

The generators $a_\alpha$ need not be independent. A relation among them is a vector $p=(p_\alpha)\in R^N$ satisfying
\[
\sum_{\alpha=1}^N p_\alpha a_\alpha=0,
\qquad\text{equivalently}\qquad \widehat H p=0.
\]
Thus these relations are precisely the compactly supported zero modes. If $w_j$ are the columns of $h_1$, exactness says that every such zero mode is a finite linear combination of translates of the $w_j$.

The $w_j$ may themselves be redundant. The columns of $h_2$ generate all relations $\sum_j q_jw_j=0$ with $q_j\in R$, and injectivity of $h_2$ means that there are no further relations among these chosen generators of relations. In this sense, the sequence records local generators, their dependencies, and the dependencies among those dependencies. The argument below is still needed to show that evaluation at each momentum preserves exactness.
\end{rem}

\begin{lemma}
    \label{ConstRank}
    Assume that $\mathrm{Rank}_\Cee (h_j(\vb{k}))=M_j$ is constant for all $\vb{k}$. Then the rank $h_{j+1}(\vb{k})$ is also constant and full, namely 
    $$\mathrm{Rank}_\Cee(h_{j+1}(\vb{k}))=\mathrm{Null}_\Cee(h_{j}(\vb{k}))=N_j-M_j,$$
    where $\mathrm{Null}_\Cee(h_{j}(\vb{k}))$ is the dimension of the kernel of $h_{j}(\vb{k})$.
\end{lemma}
Since $\mathrm{Im}_R(h_{j+1})= \Ker_R(h_j) \implies \mathrm{Im}_\Cee(h_{j+1}(\vb{k})) \subseteq \Ker_\Cee(h_j(\vb{k})) $, the important conclusion of this lemma is that the proper inclusion, which would have $\text{Rank}_\Cee(h_{j+1}(\vb{k})) < \text{Null}_\Cee(h_j(\vb{k}))$, does not happens. We'll delay the proof of this lemma. 

By the gap assumption, $\ft{H}(\vb{k})$ does have constant rank $M_0=M$. Applying the lemma inductively, we conclude that $\text{Rank}_\Cee(h_{j+1}(\vb{k}))=\text{Null}_\Cee(h_{j}(\vb{k}))=N_j-M_j$ for all $j$, and hence we have the following exact sequence of vector spaces:
\begin{lemma}
     \label{CElemma}
     The sequence of $\Cee$-vector spaces obtained by the evaluation of~(\ref{Syzygy}) at \( \forall \vb{k} \in \Tee^{2} \)
   \begin{equation}
     0 \xrightarrow{h_3(\vb{k})} \Cee^{N_2} \xrightarrow{h_2(\vb{k})} \Cee^{N_1} \xrightarrow{h_{1}(\vb{k})} \Cee^N \xrightarrow{h_{0}(\vb{k})} \mathrm{Im}_{\Cee}(\widehat{H}(\vb{k})) \to 0 \label{CExact}
   \end{equation}
is exact.
\end{lemma}

With this, we can prove the triviality of $E$ and hence the theorem. By~\eqref{CExact}, $\Ker_\Cee(h_2(\vb{k}))= \{0\}$ for all $\vb{k}$ and so the columns of $h_2(\vb{k})$ are linearly independent at each $\vb{k}$, so they form a global smooth frame since $\vb{k}\mapsto h_2(\vb{k})$ is real-analytic. Hence $\bigsqcup_{\vb{k}\in\Tee^2}\,\mathrm{Im}_\Cee(h_2(\vb{k}))$  is a trivial bundle; in particular its first Chern number vanishes. With this,
$$\bigsqcup_{\vb{k}\in\Tee^2} \Cee^{N_1}=\bigsqcup_{\vb{k}\in\Tee^2} \,\Ker_\Cee(h_1(\vb{k})) \oplus \big(\Cee^{N_1} / \Ker_\Cee(h_1(\vb{k}))\big)=\bigsqcup_{\vb{k}\in\Tee^2} \,\mathrm{Im}_\Cee(h_2(\vb{k})) \oplus  \mathrm{Im}_\Cee(h_1(\vb{k})),$$
and the additivity of Chern number implies 
$$0=\mathrm{c}_1\bigg(\bigsqcup_k \Cee^{N_1}\bigg)=0+\mathrm{c}_1\Big(\mathrm{Im}_\Cee(h_1(\vb{k}))\Big).$$
Since the first Chern number completely classifies complex vector bundles over a closed two-dimensional manifold, we conclude that $\bigsqcup_{\vb{k}} \mathrm{Im}_\Cee(h_1(\vb{k}))$ is trivial. Repeating this argument but starting with $h_1$ yields that $E=\bigsqcup_{\vb{k}\in\Tee^2}\, \mathrm{Im}_\Cee(\widehat{H}(\vb{k})) =\bigsqcup_{\vb{k}\in\Tee^2}\, \mathrm{Im}_\Cee(h_0(\vb{k}))$ is trivial. \teop
\vspace{4mm}

\subsection{Proof of Lemma \ref{ConstRank}}

Let $\bbF$ be the field of fractions of $R$.

We note that if $p\in R$ is such that $p(\vb{k}) = 0$ for all $\vb{k}\in \Tee^2$, then $p=0$ in $R$. Let $A\in M_{p\times q}(R)$. If $\mathrm{Rank}_\bbC(A(\vb{k}))=s$, then some $s\times s$ minor of $A(\vb{k})$ is
nonzero. That number is the evaluation of the corresponding $s\times s$
minor of $A$, which is therefore a nonzero element of $R$; hence $\mathrm{Rank}_\bbF(A)\ge s$, namely
\begin{equation*}
\mathrm{Rank}_\bbC(A(\vb{k}))\leq \mathrm{Rank}_\bbF(A).
\end{equation*}
Moreover, if $\mathrm{Rank}_\bbC(A(\vb{k}))\leq r$ for all $\vb{k}\in\Tee^2$, then every $(r+1)\times(r+1)$ minor of $A$ is an element of $R$ whose evaluation vanishes at
every point of $\bbT^2$. By opening remark, it is the zero element of $R$. Hence $\mathrm{Rank}_\bbF(A)\le r$. Summarizing this,

\begin{prop}\label{prop:rank}
Let $A\in M_{p\times q}(R)$ and suppose that $\vb{k}\mapsto\mathrm{Rank}_\bbC(A(\vb{k}))$
is constant on $\bbT^2$. Then
\begin{equation*}
\mathrm{Rank}_\bbF(A) = \mathrm{Rank}_\bbC\bigl(A(\vb{k})\bigr)
   \qquad\text{for every }\vb{k}\in\bbT^2 .
\end{equation*}
\end{prop}

\begin{cor}\label{cor:H}
Let $H$ be a gapped translation-invariant $\mathrm F^4$ Hamiltonian and $M=\mathrm{Rank}_\bbC
\ft H(\vb{k})$. Then
$\mathrm{Rank}_\bbF(h_0)=\mathrm{Rank}_\bbF(\ft H)=M$.
\end{cor}
\begin{proof}
Immediate from Proposition~\ref{prop:rank} since the rank of $\ft H(\vb{k})$ is constant by the gap assumption.
\end{proof}

A priori, exactness of (5.4) only gives $\operatorname{Im}_\bbC(h_1(\vb{k}))\subseteq
\Ker_\bbC(\ft{H}(\vb{k}))$, that is
\begin{equation}\label{eq:apriori}
   \mathrm{Rank}_\bbC\bigl(h_1(\vb{k})\bigr)\;\le\;\mathrm{Null}_\bbC\bigl(\ft{H}(\vb{k})\bigr)\;=\;N-M .
\end{equation}
In order to prove that the inequality is in fact an equality, we now construct $N-M$ linearly independent vectors of $\Ker_\bbC(\ft{H}(\vb{k}))$ that are evaluations of
elements of $\Ker_R(\ft{H})$.

Fix $\vb{k}\in\bbT^2$. Since $\mathrm{Rank}_\bbC(\ft{H}(\vb{k}))=M$, there are (by row reduction) an invertible $V\in M_{N\times
N}(\bbC)$ and a permutation matrix $P\in M_{N\times N}(\bbC)$ such that
\begin{equation}\label{eq:one}
   V\ft{H}(\vb{k})P =
   \begin{pmatrix}\mathbb I_M & g(\vb{k})\\[2pt] 0 & 0\end{pmatrix}.
\end{equation}
Note that the permutation is needed here to bring the $M$ linearly independent columns of $\ft{H}(\vb{k})$ to the left. Applying the same $V$ and $P$ to $\ft{H}$ as a matrix over $R$ and splitting into the same blocks,
\begin{equation}\label{eq:two}
   V\ft{H}P =
   \begin{pmatrix} i & g\\ f & j\end{pmatrix},
   \qquad i\in M_{M\times M}(R),
\end{equation}
so that \eqref{eq:one} reads $i(\vb{k})=\mathbb I_M$, $f(\vb{k})=0$, $j(\vb{k})=0$. In particular,
\[
   \Delta=\det i\in R \quad\text{satisfies}\quad \Delta(\vb{k})=\det\mathbb I_M=1 ,
\]
hence $\Delta\ne 0$ in $R$ and $i$ is invertible over $\bbF$. Row reduction over $\bbF$ then
gives, with $G=i^{-1}g$ and
$U=\left(\begin{smallmatrix} i^{-1}&0\\ -f\,i^{-1}&\mathbb I\end{smallmatrix}\right)
\in\mathrm{GL}_N(\bbF)$,
\begin{equation}\label{eq:three}
   UV\ft{H}P =
   \begin{pmatrix}\mathbb I_M & G\\[2pt] 0 & \;j-f\,i^{-1}g\end{pmatrix} .
\end{equation}
Since $U,V,P$ are invertible, the rank over $\bbF$ of the right-hand side equals
$\mathrm{Rank}_\bbF(\ft{H})$, which is $M$ by Corollary~\eqref{cor:H}. The first $M$ rows of \eqref{eq:three} are
already $\bbF$-independent, hence $\mathrm{Rank}_\bbF(j-f\,i^{-1}g)=0$, that is $j-f\,i^{-1}g=0$ and 
\begin{equation}\label{eq:four}
   UV\ft{H}P=\begin{pmatrix}\mathbb I_M & G\\ 0&0\end{pmatrix}.
\end{equation}

Let $e_1,\dots,e_{N-M}$ be the standard basis of $R^{N-M}$ and set for each $l=1,\dots,N-M$, 
\[
   W_l = \Delta P\begin{pmatrix} G e_l\\ -e_l\end{pmatrix}.
\]
which are elements of $R^N$. By \eqref{eq:four},
\[
   UV\ft{H}P\begin{pmatrix} G e_l\\ -e_l\end{pmatrix} =0 ,
\]
and since $U,V$ are invertible, $\ft{H}\,W_l=0$, namely $W_l\in\Ker_R(\ft{H})$.

Evaluating at $\vb{k}$ and using $\Delta(\vb{k})=1$, $i(\vb{k})=\mathbb I_M$, hence $G(\vb{k})=g(\vb{k})$,
\[
   W_l(\vb{k})\;=\;P\begin{pmatrix} g(\vb{k})e_l\\ -e_l\end{pmatrix},
\]
and $\{W_l(\vb{k})\}_{l=1}^{N-M}$ are $\bbC$-linearly independent because their lower blocks
$-e_l$ are, and $P$ is invertible. These are the $N-M$ linearly independent elements in $\Ker_\Cee(\ft H(\vb{k}))$.

Finally, by exactness of (5.4), $\Ker_R(\ft{H})=\operatorname{Im}_R(h_1)$, so for each $l$, there is $\psi_l\in R^{N_1}$ such that $W_l=h_1\psi_l$, and in turn $W_l(\vb{k})=h_1(\vb{k})\psi_l(\vb{k})\in
\mathrm{Im}_\bbC(h_1(\vb{k}))$. Hence $\mathrm{Rank}_\bbC(h_1(\vb{k}))\ge N-M$, and with the a priori
bound~\eqref{eq:apriori},
\[
   \mathrm{Rank}_\bbC\bigl(h_1(\vb{k})\bigr)\;=\;N-M \qquad\text{for every }\vb{k}\in\bbT^2 .
\]
This is Lemma~\ref{ConstRank} for $j=0$.

The argument used only that $h_0=\ft{H}$ is a matrix with entries in $R$ whose pointwise rank is constant on $\bbT^2$. Since Proposition~\ref{prop:rank} is stated for an arbitrary rectangular matrix over $R$, it applies similalry to $h_j\in M_{N_{j-1}\times N_j}(R)$, with $M$ replaced by $M_j$ and
$N$ by $N_j$ throughout \eqref{eq:one}--\eqref{eq:four}. The lemma follows by induction on $j$.
\qed

\bigskip

\noindent{\bf Acknowledgments.} MF and SB would like to thank Wojciech De Roeck for many discussions on this and related topics. SB acknowledges financial support from NSREC of Canada. M.F. and Z.D were supported by the NSF under grant DMS-2407290.

\bigskip

\noindent{\bf AI Disclosure.} AI was used during proofreading.

\bibliographystyle{plain}
\bibliography{references.bib}

\end{document}